\documentclass[letterpaper, 10 pt, conference]{ieeeconf}  

\IEEEoverridecommandlockouts                              

\usepackage{amsfonts}
\usepackage{amsmath}
\usepackage{ifthen}
\usepackage{fancyhdr}
\usepackage{amsfonts}
\usepackage{graphicx}
\usepackage{amsmath}
\usepackage{ifthen}
\usepackage{color}
\usepackage{indentfirst}
\usepackage{epstopdf}

\DeclareMathOperator{\sgn}{sgn}
\newcommand{\E}{\mathbb{E}}

\newcommand{\Prob}{\mathbb{P}}
\makeatletter
\newcommand*{\rom}[1]{\expandafter\@slowromancap\romannumeral #1@}
\makeatother
\newtheorem{theorem}{Theorem}
\newtheorem{lemma}{Lemma}

\newtheorem{remark}{Remark}

\title{\LARGE \bf Optimal Sensor Scheduling and Remote Estimation over an Additive Noise Channel}
\author{Xiaobin Gao, Emrah Akyol, and Tamer Ba\c{s}ar \thanks{This research was supported in part by NSF, and in part by the U.S. Air Force Office of
Scientific Research (AFOSR) MURI grant FA9550-10-1-0573.}\thanks{All authors are with the Coordinated Science Laboratory, University of Illinois at Urbana-Champaign, Urbana, IL 61801; emails: \{xgao16, akyol, basar1\}@illinois.edu}}

\begin{document}

\maketitle
\thispagestyle{empty}
\pagestyle{empty}

\begin{abstract}
We consider a sensor scheduling and remote estimation problem with one sensor and one estimator. At each time step, the sensor makes an observation on the state of a source, and then decides whether to transmit its observation to the estimator or not. The sensor is charged a cost for each transmission. The remote estimator generates a real-time estimate on the state of the source based on the messages received from the sensor.  The estimator is charged for estimation error. As compared with previous works from the literature, we further assume that there is an additive communication channel noise. As a consequence, the sensor needs to encode the message before transmitting it to the estimator. For some specific distributions of the underlying random variables, we obtain the optimal solution to the problem of minimizing the expected value of the sum of communication cost and estimation cost over the time horizon.
\end{abstract}

\section{Introduction}
The sensor scheduling and remote state estimation problem arises in the applications of wireless sensor networks, such as environmental monitoring and networked control systems. As an example of environmental monitoring, people in the National Aeronautics and Space Administration (NASA) Earth Science group want to monitor the evolution of the soil moisture, which is used in weather forecast, ecosystem process simulation, etc \cite{NASA10}. In order to achieve that goal, the sensor networks are built over an area of interest. The sensors collect data on the soil moisture and send them to the decision unit at NASA via wireless communication. The decision unit at NASA forms estimates on the evolution of the soil moisture based on the messages received from the sensors. Similarly, in networked control systems, the objective is to control some remote plants. Sensor networks are built to measure the states of the remote plants, which then transmit their measurements to the controller via a wireless communication network. The controller estimates the state of the remote plant and generates the control signal based on that estimate \cite{Hespanha07}. In both scenarios, the quality of the remote state estimation strongly affects the quality of decision making at the remote site, that is, weather prediction or control signal generation. The networked sensors are usually constrained by limits on energy. They are not able to communicate with the estimator at every time step and thus, the estimator has to produce its best estimate based on the partial information received from the sensors. Therefore, the communication between the sensors and the estimator should be scheduled wisely, and the estimator should be designed properly, so that the state estimation error is minimized subject to the communication constraints.

The sensor scheduling and remote state estimation problem has been extensively studied in recent years. Imer and Ba\c{s}ar \cite{Imer10} considered the model where the sensor is allowed to communicate with the estimator only a limited number of times. The continuous-time version of the problem in \cite{Imer10} has been studied by Rabi \emph{et al.} \cite{Rabi06}. Xu and Hespanha \cite{Xu04} considered the networked control problem involving the state estimation and communication scheduling, which can be viewed as the sensor scheduling and remote estimation problem. Wu \emph{et al.} \cite{Wu13} considered the sensor scheduling and estimation problem subject to constraints on the average communication rate over the infinite-time horizon, which can also be viewed as Kalman-filtering with scheduled observations. 
Lipsa and Martins \cite{Lipsa11} considered the setup where the sensor is not constrained by communication times but is charged a communication cost. Nayyar \emph{et al.} \cite{Nayya13} considered a similar problem with an energy harvesting sensor.

In previous works, the communication between the sensor and the estimator has been assumed to be perfect (no additive channel noise), which may not be realistic, even though it was an important first step.  This paper investigates the effect of communication channel noise on the design of optimal sensor scheduling and remote estimation strategies. We consider a discrete time sensor scheduling and remote estimation problem over a finite-time horizon, where there is one sensor and one remote estimator. We assume that at each time step, the sensor makes a perfect observation on the state of an independent identically distributed (i.i.d.) source. Next, the sensor decides whether to transmit its observation to the remote estimator or not. The sensor is charged a cost for each transmission (communication cost). Since the communication channel is noisy, the sensor encodes the message before transmitting it to the estimator. The remote estimator generates real-time estimate on the state of the source based on the noise corrupted messages received from the sensor. The estimator is charged for estimation error (estimation cost).  Our goal is to design the communication scheduling strategy and encoding strategy for the sensor, and the estimation strategy (decoding strategy) for the estimator, to minimize the expected value of the sum of communication cost and estimation cost over the time horizon. Our solution consists of a threshold-based communication scheduling strategy, and a pair of piecewise linear encoding/decoding strategies. We show that the proposed solution is  optimal when the random variables have some specific distributions.
%

\section{Problem Formulation}
\label{ProblemFormulation}
\subsection{System Model}
\begin{figure}[h]
\centering
\includegraphics[height=24mm]{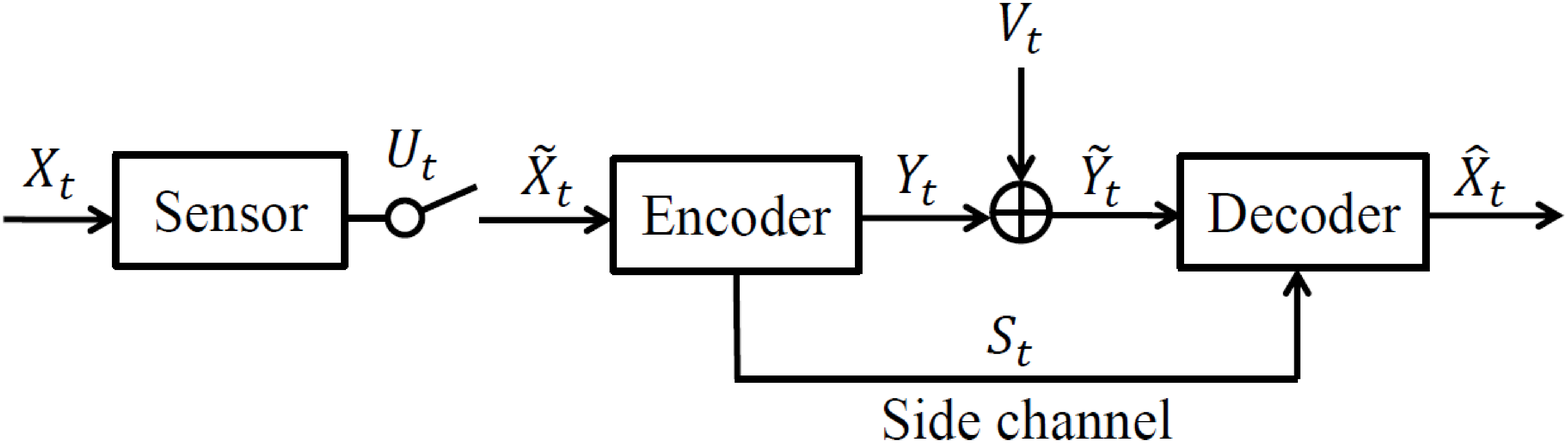}
\caption{System model}
\end{figure}
Consider a discrete time sensor scheduling and remote estimation problem over a finite-time horizon, that is, $t=1,2,\ldots,T$. In the problem, there is \emph{one} sensor, \emph{one} encoder and \emph{one} remote estimator (which is also called ``decoder"). The sensor observes an independent identically distributed (i.i.d.) stochastic process, $\{X_t\}$, $X_t\in\mathbb{R}$, which has Laplace density, $p_X$, with parameters $(0,\lambda^{-1})$. Hence,
\begin{equation*}
p_X(x) =
\begin{cases}
\frac{1}{2}\lambda e^{-\lambda x}, & \mbox{if } x\geq 0 \vspace{0.2cm}
\\ \frac{1}{2}\lambda e^{\lambda x}, & \mbox{if } x< 0
\end{cases}
\end{equation*}
Assume that at time $t$ the sensor has perfect observation on $X_t$.  Then the sensor decides whether to transmit the measurement to the encoder or not. Let $U_t\in\{0,1\}$ be the sensor's decision at time $t$, where $U_t=0$ stands for no communication and $U_t=1$ stands for communication. If the sensor communicates at time $t$, it will be charged a cost $c$. ``No communication" results in zero communication cost. Assume that the communication between the sensor and the encoder is perfect. Denote by $\tilde{X}_t$ the message received by the encoder, that is
\begin{equation*}
\tilde{X}_t =
\begin{cases}
X_t, & \mbox{if } U_t=1 \vspace{0.2cm}
\\ \epsilon, & \mbox{if } U_t=0
\end{cases}
\end{equation*}
where $\epsilon$ is a free symbol standing for no message is transmitted. Once the encoder receives the message from the sensor, it sends an encoded message to the decoder, denoted by $Y_t$. The encoder will not send any message to the decoder if it does not receive any message from the sensor, which is denoted by $Y_t=\epsilon$. Assume that the encoded message is corrupted by an additive channel noise $V_t$, $V_t\in\mathbb{R}$. $\{V_t\}$ is an i.i.d. random process, which is independent of $\{X_t\}$. We take $V_t$ to have gamma distribution $\Gamma(k,\theta)$. Denoting the message received by the decoder by $\tilde{Y}_t$, we have
\begin{equation*}
\tilde{Y}_t =
\begin{cases}
Y_t+V_t, & \mbox{if } Y_t\neq\epsilon \vspace{0.2cm}
\\ \epsilon, & \mbox{if } Y_t=\epsilon
\end{cases}
\end{equation*}
When sending the encoded message to the decoder, the encoder will transmit the sign of $\tilde{X}_t$ to the decoder via a side channel, denoted by $S_t$. Again, the encoder will not send any message to the decoder via the side channel if it does not receive any message from the sensor. Assume that the side channel is noise-free; then
\begin{equation*}
S_t =
\begin{cases}
\sgn(\tilde X_t), & \mbox{if } \tilde X_t\neq\epsilon \vspace{0.2cm}
\\ \epsilon, & \mbox{if } \tilde X_t=\epsilon
\end{cases}
\end{equation*}
After receiving $\tilde{Y}_t$ and $S_t$, the decoder produces an estimate on $X_t$, denoted by $\hat{X}_t$. The estimator will be charged for distortion in estimation. Assume that the distortion function $\rho(X_t,\hat{X}_t)$ is the squared error $(X_t-\hat{X}_t )^2$, and the encoder has average power constraint:
$$
\E[Y_t^2|U_t=1]\leq P_T,
$$
where $P_T$ is known. The cost  at time $t$ is
$$
cU_t+{(X_t-\hat{X}_t)}^2, \;\;\;c > 0,
$$
where $cU_t$ is the communication cost and $(X_t-\hat{X}_t )^2$ is the estimation cost.

\subsection{Decision Strategies}
Assume that at time $t$, the sensor has memory on all its observations up to $t$, denoted by $X_{1:t}$, and all the decisions it has made up to $t-1$, denoted by $U_{1:t-1}$. The sensor determines whether to communicate or not at time $t$ based on its current information $(X_{1:t},U_{1:t-1})$, namely
$$
U_t=f_t(X_{1:t},U_{1:t-1}),
$$
where $f_t$ is the scheduling policy of the sensor at time $t$ and $\textbf{f}=\{f_1,f_2,\ldots,f_T\}$ is the scheduling strategy of the sensor.

Similarly, at time $t$, the encoder has memory on all the messages received from the sensor up to $t$, denoted by $\tilde{X}_{1:t}$, and all the encoded messages it has sent to the decoder up to $t-1$, denoted by $Y_{1:t-1}$. The encoder generates the encoded message at time $t$ based on its current information $(\tilde{X}_{1:t},Y_{1:t-1})$, namely
$$
Y_t=g_t(\tilde{X}_{1:t},Y_{1:t-1}),
$$
where $g_t$ is the encoding policy of the encoder at time $t$ and $\textbf{g}=\{g_1,g_2,\ldots,g_T\}$ is the encoding strategy.

Finally, assume that at time $t$, the decoder has memory on all the messages received from the encoder up to $t$, denoted by $\tilde{Y}_{1:t},S_{1:t}$. The decoder generates the estimate at time $t$ based on its current information $(\tilde{Y}_{1:t},S_{1:t})$, namely
$$
\hat{X}_t=h_t(\tilde{Y}_{1:t},S_{1:t})
$$
where $h_t$ is the policy of the decoder at time $t$ and $\textbf{h}=\{h_1,h_2,\ldots,h_T\}$ is the decoding strategy.

\begin{remark}Although we do not assume that the encoder has memory on $U_{1:t},S_{1:t}$, yet it can deduce them from $\tilde{X}_{1:t}$. Similarly, the decoder can obtain $U_{1:t}$ from $\tilde{Y}_{1:t}$. \end{remark}

\subsection{Assumptions on the Parameters}
Denote by $\sigma_V^2$ the variance of $V_t$, and recall that $V_t$ has gamma distribution $\Gamma(k,\theta)$. Then, $\sigma_V^2=k\theta^2$. Define $\alpha=:\lambda\sqrt{P_T}$, and $\gamma:=\frac{P_T}{\sigma_V^2}$. Assume that
$$
\theta=\sqrt{P_T}.
$$
Then, we have
\begin{equation}
\label{assumptions on parameters}
\alpha=\lambda\theta \mbox{, } \gamma=\frac{1}{k}.
\end{equation}
\begin{remark}
These parameters as well as the form of source and channel noise distributions, even though they might seem arbitrary at first sight, are carefully chosen to render the affine encoding and decoding strategies optimal, as demonstrated later in Lemma 3. In practice, such a coincidence could be a deliberate objective of an adversarial agent that controls the channel noise density, since affine encoding/decoding strategies are worst-case optimal encoding/decoding policies under second order statistical constraints (such as encoding power) and distortion measures (such as mean squared error), as studied in \cite{akyol2013optimal}.
\end{remark}

\subsection{Optimization Problem}
Consider the system described above, given the time horizon $T$, the statistics of $\{X_t\}$ and $\{V_t\}$, the communication cost $c$, and the power constraint $P_T$. Determine the scheduling strategy, encoding strategy, and decoding strategy $(\textbf{f},\textbf{g},\textbf{h})$ for the sensor, the encoder, and the decoder, respectively, that minimize the expected value of the sum of communication and estimation costs over $T$, namely,
$$
J(\textbf{f},\textbf{g},\textbf{h})= \E \left\{\sum_{t=1}^T cU_t+{(X_t-\hat{X}_t)}^2\right\}
$$
subject to the power constraint of the encoder.

\section{Prior Work}
\label{Prior Work}
We review some results of a communication problem studied in \cite{Akyol13}. Consider the communication system described in Fig.\ref{priorWork}.
\begin{figure}[h]
\centering
\includegraphics[height=15mm]{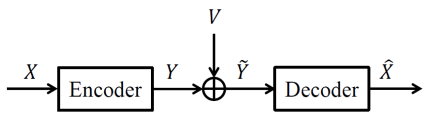}
\caption{Communication system considered in earlier work}
\label{priorWork}
\end{figure}
The encoder wants to transmit an input signal $X$ to the decoder via a communication channel. $X$ is a random variable, $X\in\mathbb{R}$. The communication channel has an additive channel noise $V$. $V$ is also a random variable, $V\in\mathbb{R}$. Assume that $X$ and $V$ are independent, and have  characteristic functions $F_{X}(\omega)$ and $F_{V}(\omega)$, respectively. Denote by $\sigma^2_{X}$ and $\sigma^2_{V}$ the variances of $X$ and $V$, respectively. Since the communication channel is noisy, the encoder encodes the message before sending it to the decoder. Assume that the encoder generates the message $Y$ according to an encoding policy $g$:
$$
Y = g(X).
$$
The encoder is constrained by average power $P_{T}$ such that
$$
\E[Y^2] \leq P_{T}.
$$
The decoder receives the noise corrupted message $Y+V$, denoted by $\tilde{Y}$, and generates an estimate of $X$, denoted by $\hat{X}$, according to some decoding policy $h$:
$$
\hat{X} = h(\tilde{Y}).
$$

\textit{Problem 1}: Given zero-mean random variables $X$ and $V$, and the power constraint of the encoder $P_T$, find  $g$ and $h$ that minimize
$$
J(g,h)=\E[(X-\hat{X})^2],
$$
subject to $\E[g^2(X)] \leq P_T$.

We reproduce the following theorem from \cite{Akyol13}.
\begin{theorem} [\cite{Akyol13}]In Problem 1,
 the optimal encoder and decoder are both linear if and only if the characteristic functions of $X$ and $V$ satisfy
    $$
    F_{X}(\alpha\omega) = F_{V}^{\gamma}(\omega),
    $$
    where $\alpha=\sqrt{\frac{P_{T}}{\sigma^2_{X}}}$ and $\gamma=\frac{P_{T}}{\sigma^2_{V}}$. Moreover, the linear encoding/decoding policies $(g^\ast,h^\ast)$ are as follows:
    \begin{equation}
    \label{prior work optimal encoding and decoding}
    \begin{array}{lcl}
    Y &=& g^\ast(X) = \alpha X, \vspace{0.2cm}
    \\ \hat{X} &=& h^\ast(\tilde{Y}) = \frac{1}{\alpha}\frac{\gamma}{\gamma+1}\tilde{Y}.
    \end{array}
    \end{equation}
\label{matching conditions}
\end{theorem}
We apply Theorem \ref{matching conditions} to obtain the following lemma.
\begin{lemma}
\label{special case of matching conditions}
Consider the communication problem described above, assume that $X$ and $V$ can be written as
$$
X = X_e - \lambda^{-1}\mbox{, }V = V_g - k\theta
$$
where $X_e$ has exponential distribution with parameter $\lambda$, and $V_g$ has gamma distribution with parameters $(k,\theta)$. Furthermore, let $\theta$ and $P_T$ satisfy $\theta=\sqrt{P_T}$. Then the optimal encoding/decoding policies $(g^{\ast},h^{\ast})$ are as described in \eqref{prior work optimal encoding and decoding} with $\alpha = \lambda\sqrt{P_T}$ and $\gamma = \frac{P_{T}}{k\theta^2}$.
\end{lemma}

\begin{proof}From the definitions of $X$ and $V$, we have $\E[X],\E[V]=0$, $\sigma^2_X=\lambda^{-2}$ and $\sigma^2_V=k\theta^2$. Then $\alpha = \sqrt{\frac{P_{T}}{\sigma^2_{X}}} = \lambda\sqrt{P_T}$ and $\gamma = \frac{P_{T}}{\sigma^2_{V}} = \frac{P_{T}}{k\theta^2}$. By the assumptions on the parameters \eqref{assumptions on parameters}
$$
\begin{array}{lcl}
F_X(\alpha\omega) &=& \E\left[e^{j\alpha\omega (X_e-\lambda^{-1})}\right] = \E\left[e^{j\alpha\omega X_e}\right]e^{-j\alpha\omega\lambda^{-1}} \vspace{0.2cm}
\\ &=& (1-j\omega\theta)^{-1}e^{-j\omega\theta}.
\end{array}
$$
Similarly,
$$
F_{V}(\omega) = \E\left[e^{j\omega (V_g-k\theta)}\right] = \left[(1-j\omega\theta)^{-1}e^{-j\omega\theta}\right]^k.
$$
Hence,
$$
F_{V}^{\gamma}(\omega) = \left[F_{V}(\omega)\right]^{\frac{1}{k}} = F_X(\alpha\omega).
$$
Applying Theorem \ref{matching conditions}, we obtain the desired result.

\end{proof}
We next extend Theorem \ref{matching conditions} to variables with non-zero means.
\begin{lemma}
\label{matching conditions transform}
Let the optimal encoding/decoding policies to \textit{Problem 1} be $(g^\ast,h^\ast)$. Consider the communication problem with non-zero mean random variables $X^\prime$ and $V^\prime$; call it \textit{Problem 2}. Assume that $X^\prime$ and $V^\prime$ are affine transforms of $X$ and $V$:
$$
X^\prime = aX+b_1 \mbox{, } V^\prime = V+b_2,
$$
where $a,b_1,b_2$ are known constants, $a\in\{-1,+1\}$, $b_1,b_2\in\mathbb{R}$. Then, the optimal encoding/decoding policies for \textit{Problem 2}, denoted by $(g^{\prime\ast},h^{\prime\ast})$, are
$$
\begin{array}{lcl}
g^{\prime\ast}(X^\prime) &=& g^\ast\left(\frac{X^\prime-b_1}{a}\right), \vspace{0.2cm}
\\ h^{\prime\ast}(\tilde{Y^\prime}) &=& a\cdot h^\ast(\tilde{Y}^\prime-b_2)+b_1.
\end{array}
$$
Moreover, the  costs of the two problems are equivalent.
\end{lemma}

\begin{proof}For any encoding/decoding policies $(g,h)$ satisfying the power constraint in \textit{Problem 1}, we have
$$
Y = g(X) \mbox{, } \E [Y^2] \leq P_T \mbox{, } \hat{X} = h(\tilde{Y}) = h(Y+V).
$$
Define the encoding/decoding policies $(g^\prime,h^\prime)$ in \textit{Problem 2} as follows:
$$
\begin{array}{lcl}
Y^\prime=g^{\prime}(X^\prime)= g\left(\frac{X^\prime-b_1}{a}\right) \vspace{0.2cm}
\\ \hat{X^\prime} = h^{\prime}(\tilde{Y^\prime})= a\cdot h(\tilde{Y}^\prime-b_2)+b_1.
\end{array}
$$
where $Y^\prime$ is the output of the encoder, $\tilde{Y^\prime}$ is the noise corrupted message received by the decoder, and $\hat{X^\prime}$ is the output of the decoder in \textit{Problem 2}. Furthermore,
$$
\tilde{Y^\prime} = Y^\prime + V^\prime,
$$
and
$$
Y^\prime = g\left(\frac{X^\prime-b_1}{a}\right) = g(X)= Y.
$$
Hence, $\E [{Y^\prime}^2] = \E [{Y}^2] \leq P_T$, which implies that the pair $(g^\prime,h^\prime)$ satisfies the power constraint of the encoder in \textit{Problem 2}. Moreover,
$$
\begin{array}{lcl}
\hat{X^\prime} &= a\cdot h(\tilde{Y}^\prime-b_2)+b_1 &= a\cdot h(Y^\prime+V^\prime-b_2)+b_1, \vspace{0.2cm}
\\ &= a\cdot h(Y+V)+b_1 &= a\hat{X}+b_1.
\end{array}
$$
Let $J_1(g,h)$ be the cost corresponding to $(g,h)$ in \textit{Problem 1} , and $J_2(g^\prime,h^\prime)$ be the cost corresponding to $(g^\prime,h^\prime)$ in \textit{Problem 2}. Then,
$$
\begin{array}{lcl}
J_2(g^\prime,h^\prime) &=& \E\left[\big(X^\prime-\hat{X^\prime}\big)^2\right] \vspace{0.2cm}
\\ &=&  \E\left[\big(aX+b_1-a\hat{X}-b_1\big)^2\right] \vspace{0.2cm}
\\ &=& a^2\cdot\E\left[(X-\hat{X})^2\right] = J_1(g,h).
\end{array}
$$
Conversely, for any encoding/decoding policies $(g^\prime,h^\prime)$ satisfying the power constraint in \textit{Problem 2}, define the encoding/decoding policies $(g,h)$ in \textit{Problem 1} as follows:
$$
\begin{array}{lcl}
Y=g(X) = g^\prime\left(aX+b_1\right) \vspace{0.2cm},
\\ \hat{X} = h(\tilde{Y}) = \frac{1}{a}[h^\prime(\tilde{Y}+b_2)-b_1].
\end{array}
$$
By the same procedure it can be shown that the pair $(g,h)$ satisfies the power constraint of the encoder in \textit{Problem 1} and $J_1(g,h)=J_2(g^\prime,h^\prime)$. Therefore, the  costs of \textit{Problem 1} and \textit{Problem 2} are equivalent. Moreover, if the optimal cost of \textit{Problem 1} is achieved by $(g^\ast,h^\ast)$, then the optimal cost of \textit{Problem 2} can be achieved by $(g^{\prime\ast},h^{\prime\ast})$.

\end{proof}

\begin{lemma}
\label{exponential source and gamma noise}
Consider the communication problem described above, assume that the input signal $X$ has exponential distribution with parameter $\lambda$, and  channel noise $V$ has gamma distribution with parameters $(k,\theta)$. Furthermore, let $\theta, P_T$ satisfy $\theta=\sqrt{P_T}$.
Then the optimal encoding/decoding policies $(g^{\ast},h^{\ast})$ are as follows:
$$
\begin{array}{lcl}
Y &=& g^{\ast}(X) = \alpha X - \alpha\lambda^{-1}, \vspace{0.2cm}
\\ \hat{X} &=& h^{\ast}(\tilde{Y}) = \frac{1}{\alpha}\frac{\gamma}{\gamma+1}\tilde{Y}+\frac{\gamma}{\gamma+1}\lambda^{-1},
\end{array}
$$
where $\alpha = \lambda\sqrt{P_T}$. Moreover, the optimal cost is
$$
J(g^{\ast},h^{\ast}) = \frac{1}{\gamma+1}\frac{1}{\lambda^2} :=m.
$$
\end{lemma}
\begin{proof}Applying Lemma \ref{matching conditions transform} to Lemma \ref{special case of matching conditions}, by letting $a=1,b_1=\lambda^{-1},b_2=k\theta$, one should obtain optimal encoding/decoding policies $(g^{\ast},h^{\ast})$ described above. Furthermore,
$$
\begin{array}{lcl}
\hat{X} &=& \frac{1}{\alpha}\frac{\gamma}{\gamma+1}\tilde{Y}+\frac{\gamma}{\gamma+1}\lambda^{-1} = \frac{1}{\alpha}\frac{\gamma}{\gamma+1}(Y+V)+\frac{\gamma}{\gamma+1}\lambda^{-1}, \vspace{0.2cm}
\\&=& \frac{1}{\alpha}\frac{\gamma}{\gamma+1}(\alpha X-\alpha\lambda^{-1}+V)+\frac{\gamma}{\gamma+1}\lambda^{-1}, \vspace{0.2cm}
\\&=& \frac{\gamma}{\gamma+1}X+\frac{1}{\alpha}\frac{\gamma}{\gamma+1}V.
\end{array}
$$
Using \eqref{assumptions on parameters}, the cost of $(g^\ast,h^\ast)$ is computed as follows:
$$
\begin{array}{lcl}
J(g^\ast,h^\ast) &=& \E[(X-\hat{X})^2] \vspace{0.2cm}
\\ &=& \E\left[\big(X-\frac{\gamma}{\gamma+1}X-\frac{1}{\alpha}\frac{\gamma}{\gamma+1}V\big)^2\right] \vspace{0.2cm}
\\ &=& \frac{1}{(\gamma+1)^2}\bigg[\E[X^2]+\frac{\gamma^2}{\alpha^2}\E[V^2]-\frac{2\gamma}{\alpha} \E[XV]\bigg] \vspace{0.2cm}
\\ &=& \frac{1}{\gamma+1}\frac{1}{\lambda^2}.
\end{array}
$$
\end{proof}

\section{Main Results}
\label{Preliminary result}
\begin{theorem}
\label{one stage problem}
Consider the problem  in Section \ref{ProblemFormulation}-D.
\begin{enumerate}
\item The optimal sensor scheduling, encoding and decoding strategies are in the forms of:
$$
U_t = f_t (X_t) \mbox{, } Y_t = g_t (\tilde{X}_t) \mbox{, } \hat{X}_t = h_t(\tilde{Y}_t,S_t).
$$
\item The optimal decision strategies are stationary, i.e., there exists $(\textbf{f}^\ast,\textbf{g}^\ast,\textbf{h}^\ast)$ minimizing
$$
J(\textbf{f},\textbf{g},\textbf{h})=\E \left\{\sum_{t=1}^T cU_t+{(X_t-\hat{X}_t)}^2\right\},
$$
where
$$
\begin{array}{lcl}
\textbf{f}^\ast = \{f_1^\ast,f_2^\ast,\ldots,f_T^\ast\},&f_1^\ast=\cdots=f_T^\ast:=f^\ast, \vspace{0.2cm}
\\ \textbf{g}^\ast = \{g_1^\ast,g_2^\ast,\ldots,g_T^\ast\},&g_1^\ast=\cdots=g_T^\ast:=g^\ast, \vspace{0.2cm}
\\ \textbf{h}^\ast = \{h_1^\ast,h_2^\ast,\ldots,h_T^\ast\},&h_1^\ast=\cdots=h_T^\ast:=h^\ast.
\end{array}
$$
\end{enumerate}
\end{theorem}

\begin{proof}At time $t=T$, we want to design $(f_T,g_T,h_T)$ to minimize
$$
J_{T_1}(f_T,g_T,h_T)=\E\left\{cU_T+{(X_T-\hat{X}_T)}^2\right\},
$$
subject to the power constraint of the encoder, called \textit{Problem T1}. Let $I_{sT},I_{eT},I_{dT}$ be the information about the past system states available to the sensor, the encoder, and the decoder, respectively, at time $T$, i.e.,  $I_{sT} = \{X_{1:T-1},U_{1:T-1}\}$, $I_{eT} = \{\tilde{X}_{1:T-1},Y_{1:T-1}\}$, and $I_{dT} = \{\tilde{Y}_{1:T-1},S_{1:T-1}\}$. Recall that the decisions at time $T$ are generated by $U_T=f_T(X_T,I_{sT})$, $Y_T=g_T(\tilde{X}_T,I_{eT})$, $\hat{X}_T=h_T(\tilde{Y}_T,S_T,I_{dT})$.

Let $I_T$ be the information set about the past system states at time $T$:
$$
I_T=\{X_{1:T-1},U_{1:T-1},\tilde{X}_{1:T-1},Y_{1:T-1},\tilde{Y}_{1:T-1},S_{1:T-1}\}.
$$
Then $I_{sT},I_{eT},I_{dT}\in I_T$. Consider another problem, called \textit{Problem T2}, where $I_T$ is available to the sensor, the encoder, and the decoder, and we want to design $(f_T^{\prime},g_T^{\prime},h_T^{\prime})$ to minimize
$$
J_{T_2}(f_T^\prime,g_T^\prime,h_T^\prime)=\E\left\{cU_T+{(X_T-\hat{X}_T)}^2\right\}.
$$
subject to the power constraint of the encoder, where $U_T=f_T^\prime(X_T,I_T)$, $Y_T=g_T^\prime(\tilde{X}_T,I_T)$, $\hat{X}_T=  h_T^\prime(\tilde{Y}_T,S_T,I_T)$. Since the sensor, the encoder, and the decoder can always ignore the redundant information and behave as if they only know $I_{sT},I_{eT},I_{dT}$, respectively, the performance of system in \textit{Problem T2} is no worse than that of the \textit{Problem T1}, i.e.,
$$
\underset{(f_T^\prime,g_T^\prime,h_T^\prime)}{\min}J_{T_2}(f_T^\prime,g_T^\prime,h_T^\prime) \leq \underset{(f_T,g_T,h_T)}{\min}J_{T_1}(f_T,g_T,h_T).
$$
Similarly, consider a third problem, call it \textit{Problem T3}, where $I_{sT},I_{eT},I_{dT}$ are not available to the sensor, the encoder, and the decoder, respectively. We want to design $(f_T^{\prime\prime},g_T^{\prime\prime},h_T^{\prime\prime})$ to minimize
$$
J_{T_3}(f_T^{\prime\prime},g_T^{\prime\prime},h_T^{\prime\prime})=\E\left\{cU_T+{(X_T-\hat{X}_T)}^2\right\},
$$
subject to the power constraint of the encoder, where $U_T=f_T^{\prime\prime}(X_T)$, $Y_T=g_T^{\prime\prime}(\tilde{X}_T)$, $\hat{X}_T= h_T^{\prime\prime}(\tilde{Y}_T,S_T)$. By a similar argument as above, the system in \textit{Problem T1} cannot perform worse than the system in \textit{Problem T3}, namely,
$$
\underset{(f_T,g_T,h_T)}{\min}J_{T_1}(f_T,g_T,h_T) \leq \underset{(f_T^{\prime\prime},g_T^{\prime\prime},h_T^{\prime\prime})}{\min}J_{T_3}(f_T^{\prime\prime},g_T^{\prime\prime},h_T^{\prime\prime}).
$$
Let us come back to \textit{Problem T2}. Since the communication cost $c$, the distortion function $\rho(\cdot,\cdot)$, and the power constraint of the encoder do not depend on $I_T$,
and since $\{X_t\}$ and $\{V_t\}$ are i.i.d. random processes, $X_T$ and $V_T$ are  also independent of $I_T$, and
there is no loss of optimality by restricting
$$
U_T = f_T^\prime(X_T) \mbox{, } Y_T = g_T^\prime(\tilde{X}_T) \mbox{, } \hat{X}_T = h_T^\prime(\tilde{Y}_T,S_T),
$$
and
$$
\underset{(f_T^\prime,g_T^\prime,h_T^\prime)}{\min}J_{T_2}(f_T^\prime,g_T^\prime,h_T^\prime) = \underset{(f_T^{\prime\prime},g_T^{\prime\prime},h_T^{\prime\prime})}{\min}J_{T_3}(f_T^{\prime\prime},g_T^{\prime\prime},h_T^{\prime\prime})
$$
The equality above indicates in \textit{Problem T1} the sensor, the encoder and the decoder can ignore their information about the past, namely $I_{sT}$, $I_{eT}$, and $I_{dT}$, respectively, and there is no loss of optimality by restricting
$$
U_T = f_T(X_T) \mbox{, } Y_T = g_T(\tilde{X}_T) \mbox{, } \hat{X}_T = h_T(\tilde{Y}_T,S_T),
$$
which proves (1).

Since $(f_T,g_T,h_T)$ does not take $I_T$ as a parameter, the design of  $(f_T,g_T,h_T)$ is independent of the design of $(f_{1:T-1},g_{1:T-1},h_{1:T-1})$. Hence the problem can be viewed as a $T-1$ stages problem and a one stage problem.

By induction, we can show that the design of $(f_1,g_1,h_1)$, $(f_2,g_2,h_2)$, $\dots$, $(f_T,g_T,h_T )$ are mutually independent, where $(f_t,g_t,h_t )$ is designed to minimize
$$
J(f_t,g_t,h_t)=\E\left\{cU_t+{(X_t-\hat{X}_t)}^2\right\},
$$
subject to the power constraint of the encoder. Furthermore, since ${X_t}$ and ${V_t}$ are i.i.d. random processes. The optimal $(f_t,g_t,h_t)$ should be the same for all $t=1,2,\ldots,T$, which proves (2).

\end{proof}

By Theorem \ref{one stage problem}, the problem can be reduced to the ``one-stage" problem and the objective is to determine $(f^\ast,g^\ast,h^\ast)$. Therefore for simplicity we suppress the subscript for time in all the expressions for the rest of the paper.

\begin{theorem}
\label{linear encoding and decoding corresponding to threshold based policy}
Consider the problem formulated in Section \ref{ProblemFormulation}-D. If the sensor applies symmetric threshold-based scheduling policy $f$ as follows:
$$
U=f(X)=
\begin{cases}
1, & \mbox{if } |X|>\beta \vspace{0.2cm}
\\ 0, & \mbox{if } |X|\leq\beta
\end{cases}
$$
where $\beta$ is called the threshold, then the encoding/decoding policies $(g,h)$ described below are jointly optimal corresponding to $f$:
$$
\begin{array}{lcr}
g(\tilde{X})\;\;\;\;=
\begin{cases}
\alpha|\tilde{X}|-\alpha\beta-\alpha\lambda^{-1}, & \mbox{if } \tilde{X} \neq \epsilon \vspace{0.2cm}
\\ \epsilon, & \mbox{if } \tilde{X} = \epsilon
\end{cases}
\vspace{0.2cm}
\\ h(\tilde{Y},S)=
\begin{cases}
S\cdot\left(\frac{1}{\alpha}\frac{\gamma}{\gamma+1}\tilde{Y}+\frac{\gamma}{\gamma+1}\lambda^{-1}+\beta\right),  \mbox{if } \tilde{Y},S \neq \epsilon \vspace{0.2cm}
\\ 0, \;\;\;\;\;\;\;\;\;\;\;\;\;\;\;\;\;\;\;\;\;\;\;\;\;\;\;\;\;\;\;\;\;\;\;\;\;\;\;\;\;\;\; \mbox{if } \tilde{Y},S = \epsilon
\end{cases}
\end{array}
$$
where $\alpha = \lambda\sqrt{P_T}$, $\gamma=\frac{P_T}{k\theta^2}$.
\end{theorem}

\begin{proof}
Case \rom{1}. $U=0$, $\tilde{X},\tilde{Y},S=\epsilon$. The optimal estimator should be the conditional mean:
$$
\hat{X}=\E[X|U=0]=\E[X||X|\leq\beta]=0,
$$
where the third equality is due to the fact that $X$ is symmetrically distributed.

Case \rom{2}. $U=1$, $S=+$. The problem collapses to Problem 2. Conditioned on $\tilde{X}=X>\beta$ (which is equivalent to $U=1$, $S=+$), $\tilde{X}$ is distributed (has density) as follows:
\begin{equation*}
p_{\tilde{X}}(x) =
\begin{cases}
\lambda e^{-\lambda (x-\beta)}, & \mbox{if } x\geq \beta \vspace{0.2cm}
\\ 0, & \mbox{if } x<\beta
\end{cases}
\end{equation*}
Hence, $\tilde{X}$ can be written as $\tilde{X} = X_e + \beta$ where $X_e$ has exponential distribution with parameter $\lambda$. Since $X$ and $V$ are independent, $X>\beta$ does not affect the distribution of $V$. Therefore, $V$ has gamma distribution with parameters $(k,\theta)$. Moreover, the power constraint $P_T$ satisfies $\theta=\sqrt{P_T}$. Hence by applying Lemmas \ref{matching conditions transform} and \ref{exponential source and gamma noise}, we have the following optimal encoding/decoding policies and  cost
$$
\begin{array}{lcl}
Y &=& \alpha\tilde{X}-\alpha\beta-\alpha\lambda^{-1}, \vspace{0.2cm}
\\\hat{X} &=&\frac{1}{\alpha}\frac{\gamma}{\gamma+1}\tilde{Y}+\frac{\gamma}{\gamma+1}\lambda^{-1}+\beta.
\end{array}
$$
\begin{equation}
\label{optimal cost in case II}
\E [(X-\hat{X})^2|X > \beta]  = \frac{1}{\gamma+1}\frac{1}{\lambda^2} = m.
\end{equation}
Case \rom{3}. $U=1$, $S=-$. Again, the problem collapses to Problem 2. $U=1$, $S=-$ is equivalent to $\tilde{X}=X<-\beta$. Conditioned on that, $\tilde{X}$ can be written as
$$
\tilde{X} = -X_e - \beta,
$$
where $X_e$ has exponential distribution with parameter $\lambda$. Following the steps in Case II, we have
$$
\begin{array}{lcl}
Y &=& -\alpha\tilde{X}-\alpha\beta-\alpha\lambda^{-1} \vspace{0.2cm}
\\ \hat{X} &=& -\frac{1}{\alpha}\frac{\gamma}{\gamma+1}\tilde{Y}-\frac{\gamma}{\gamma+1}\lambda^{-1}-\beta
\end{array}
$$
\begin{equation}
\label{optimal cost in case III}
\E [(X-\hat{X})^2|X < -\beta]  = \frac{1}{\gamma+1}\frac{1}{\lambda^2} = m.
\end{equation}

\end{proof}

\begin{theorem}
\label{unique beta}
Consider the problem formulated in Section \ref{ProblemFormulation}-D. Suppose that the sensor is restricted to apply the symmetric threshold-based scheduling policy $f$ with threshold $\beta$, $\beta\in(0,+\infty)$, and the encoder/decoder apply the corresponding optimal encoding/decoding policies $(g,h)$  in Theorem \ref{linear encoding and decoding corresponding to threshold based policy}. Then, there exists a unique threshold $\beta^\ast$ minimizing the cost function among all the thresholds. Furthermore, $\beta^\ast=\sqrt{c+m}$, $m=\frac{1}{\gamma+1}\frac{1}{\lambda^2}$.
\end{theorem}

\begin{proof} The cost function subject to $f$ with threshold $\beta$, $g$, and $h$ can be written as
$$
\begin{array}{lcl}
&\;&J(f,g,h) = \E\left[cU+{(X-\hat{X})}^2\right] \vspace{0.2cm}
\\ &=& \E\left[cU+{(X-\hat{X})}^2||X|\leq\beta\right]\cdot \Prob(|X|\leq\beta) \vspace{0.2cm}
\\ &+& \E\left[cU+{(X-\hat{X})}^2|X>\beta\right]\cdot \Prob(X>\beta) \vspace{0.2cm}
\\ &+& \E\left[cU+{(X-\hat{X})}^2|X<-\beta\right]\cdot \Prob(X<-\beta)
\end{array}
$$
Consider the expectation in the first term:
$$
\begin{array}{lcl}
&\;&\E\left[cU+{(X-\hat{X})}^2||X|\leq\beta\right], \vspace{0.2cm}
\\ &=& \E\left[X^2||X|\leq\beta\right] = \int_{-\beta}^\beta x^2p_X(x)\frac{1}{\Prob(|X|\leq\beta)}dx.
\end{array}
$$
Next, consider the expectation in the second term:
$$
\begin{array}{lcl}
&\;&\E\left[cU+{(X-\hat{X})}^2|X>\beta\right] \vspace{0.2cm}
\\ &=& c+\E\left[{(X-\hat{X})}^2|X>\beta\right] = c+m,
\end{array}
$$
where the last equality is due to \eqref{optimal cost in case II}. Similarly, by  \eqref{optimal cost in case III}, the expectation in the third term is $c+m$. Hence,
$$
\begin{array}{lcl}
J(f,g,h) = \int_{-\beta}^\beta x^2p_X(x)dx \vspace{0.2 cm}
\\+(c+m)\int_{\beta}^{\infty}p_X (x)dx+(c+m)\int_{-\infty}^{-\beta}p_X (x)dx \vspace{0.2 cm}
\\ =2\int_{0}^\beta x^2p_X(x)dx+2(c+m)\int_{\beta}^{\infty}p_X (x)dx.
\end{array}
$$
Taking derivative of $J(f,g,h)$ with respect to $\beta$,
$$
\frac{d}{d\beta}J(f,g,h)=2\beta^2p_X(\beta)-2(c+m)p_X(\beta).
$$
Since $p_X(\beta)>0$, we have
$$
\frac{d}{d\beta}J(f,g,h)=
\begin{cases}
0,\;\;\;& \mbox{if } \beta=\sqrt{c+m}:=\beta^\ast \vspace{0.2cm}
\\<0,\;\;\;& \mbox{if } \beta\in(0,\beta^\ast) \vspace{0.2cm}
\\>0,\;\;\;& \mbox{if } \beta\in(\beta^\ast,\infty)
\end{cases}
$$
Hence, $\beta^\ast$ is the unique global minimizer.
\end{proof}

\section{Conclusions}
\label{Conclusions}
In this paper, we considered a sensor scheduling and remote estimation problem with noisy communication channel between the sensor and the estimator. By making specific assumptions on the distributions of source and noise, we have obtained a solution which consists of a symmetric threshold-based sensor scheduling strategy and a pair of piecewise linear encoding/decoding strategies. The study on the notion of optimality of our approach in the practical, adversarial settings where channel noise is generated by a jammer, is left to an extended version of this paper due to space constraints.  Future directions for research include extensions to hard communication cost constraints and multi-dimensional settings.

\bibliographystyle{unsrt}
\bibliography{references}

\end{document}